\documentclass[12pt,a4paper]{article}

\usepackage[utf8]{inputenc} 
\usepackage[T1]{fontenc}
\usepackage{url}
\usepackage{ifthen}
\usepackage{cite}
\usepackage[cmex10]{amsmath}
\usepackage{amssymb}
\usepackage{bm}
\usepackage{amsthm,amsfonts,ascmac}
\usepackage{cases}
\usepackage{mathtools}
\usepackage[margin=25truemm]{geometry}
\usepackage{multirow}
\usepackage[pdftex]{graphicx}

\theoremstyle{plain}

\newtheorem{thm}{Theorem}[section]
\newtheorem{dfn}[thm]{Definition}

\begin{document}
\title{Necessary and sufficient condition for constructing\\
a single qudit insertion/deletion code\\
 and its decoding algorithm}

\author{
Taro Shibayama \thanks{
Kaijo Junior and Senior High School, 
3-6-1 Okubo, Shinjuku, Tokyo, Japan, 169-0072. 
E-mail: shibayama@kaijo.ed.jp
}
}

\date{}
\maketitle

\begin{abstract}
This paper shows that Knill-Laflamme condition, known as a necessary and sufficient condition for quantum error-correction, can be applied to quantum errors where the number of particles changes before and after the error.
This fact shows that correctabilities of single deletion errors and single insertion errors are equivalent.
By applying Knill-Laflamme condition, we generalize the previously known correction conditions for single insertion and deletion errors to necessary and sufficient level.
By giving an example that satisfies this condition, we construct a new single qudit insertion/deletion code and explain its decoding algorithm.
\end{abstract}

\allowdisplaybreaks

%%%%%%%%%%%%%%%%%%%%%%%%%%%%%%%%%%%%%%%%%%%%%%%%%%%%%%%%%%%%%%
\section{Introduction}\label{sec1}
Quantum error-correcting codes play an important role in quantum information theory and have been actively studied since the 1990s\cite{Peter1995,Robert1996,Gottesman1997}.
The necessary and sufficient condition for quantum error-correction (KL condition) given by Knill-Laflamme in 1997\cite{knill1997} is extremely useful.
Quantum errors in KL condition are mainly assumed to be those in which the number of particles does not change before and after the error, such as errors represented by unitary matrices.
Recently, however, quantum errors where the number of particles changes, such as quantum insertion/deletion errors, have been attracting attention.
Since 2020, several examples of quantum insertion/deletion error-correcting codes have been reported\cite{Nakayama20201,Hagiwara20201,shibayama2021construction,9518078,9517870,ManabuHagiwara20212020XBL0191,nakamura2024decodingalgorithmcorrectingsingleinsertion}, and applications of quantum deletion codes have also been reported recently\cite{aoki2024}.

Classical insertion/deletion error-correcting codes were first given by Levenstein in 1966\cite{Levenshtein1966} and have been actively studied in recent years since then\cite{971760,TiloLeonid2013,8849583,9834350,9770830}.
The most important property of classical codes is the equivalence between the error-correctability of insertion and deletion errors\cite{Levenshtein1966,971760}.
Although it has been an open problem whether the equivalence holds in quantum theory, it has recently been shown that KL condition for deletions is equivalent to that for insertions of separable states\cite{9611450}.
In this paper, we first prove that KL condition can be used as a necessary and sufficient condition for error-correction even when the number of particles changes before and after the error.

By using KL conditions, the single deletion error-correction condition by Nakayama-Hagiwara\cite{Nakayama20202} and the single insertion error-correction condition by Shibayama-Hagiwara\cite{9834635} are improved to reach the necessary and sufficient level.
These are also greatly generalized in that they give construction conditions not only for binary codes but also for non-binary codes.
We also construct a new single qudit insertion/deletion code by giving an example that satisfies these conditions, and explain its decoding algorithm in detail.

This paper is organized as follows.
Section \ref{sec2} describes the variables and notations used in this paper.
Section \ref{sec3} explains that KL condition is also available for errors that change the number of particles.
In Section \ref{sec4}, we define single deletion and single insertion errors and show that their correctabilities are equivalent.
Section \ref{sec5} gives necessary and sufficient conditions for the correction of single insertion/deletion errors.
In Section \ref{sec6}, we give an example of the code and explain its decoding algorithm.
Finally, Section \ref{sec7} summarizes.

%%%%%%%%%%%%%%%%%%%%%%%%%%%%%%%%%%%%%%%%%%%%%%%%%%%%%%%%%%%%%
\section{Preliminaries}\label{sec2}
The symbols and notations defined in this section will be used throughout this paper.
Let $n$ be a positive integer and $[n]:=\{1,2,\dots,n\}$.
We denote by $\mathbb{C}^l$ the $l$-dimensional vector space over a complex field $\mathbb{C}$ for an integer $l\geq2$.
Let $|0\rangle,|1\rangle,|2\rangle,\dots,|l-1\rangle$ be the standard orthogonal basis of $\mathbb{C}^l$ and $\mathcal{L}:=\{0,1,2,\dots,l-1\}$.
Set $|\bm{x}\rangle:=|x_1\rangle\otimes|x_2\rangle\otimes\dots\otimes|x_n\rangle\in\mathbb{C}^{l\otimes n}$ for an $n$-tuple $\bm{x}=x_1x_2\dots x_n\in\mathcal{L}^n$.
Here $\otimes$ is the tensor product operation.
Let $\langle\bm{x}|:=|\bm{x}\rangle^\dag$ denote the conjugate transpose of $|\bm{x}\rangle$.
A positive semi-definite Hermitian matrix of trace $1$ is called a density matrix.
We denote by $S(\mathbb{C}^{l\otimes n})$ the set of all density matrices of order $l^n$.
An element of $S(\mathbb{C}^{l\otimes n})$ is called an $n$-qudit quantum state.
In this paper, we also use a complex vector $|\phi\rangle\in \mathbb{C}^{l\otimes n}$ for representing a pure state $|\phi\rangle\langle\phi|\in S(\mathbb{C}^{l\otimes n})$.

For vectors $|0_L\rangle,|1_L\rangle,|2_L\rangle,\dots,|l-1_L\rangle\in\mathbb{C}^{l\otimes n}$ of length $1$ orthogonal to each other, the image of a linear map ${\rm Enc}$ that satisfies
\begin{align*}
|\psi\rangle&:=\sum_{i\in\mathcal{L}}\alpha_i|i\rangle\in\mathbb{C}^l,\\
|\Psi\rangle&:=\sum_{i\in\mathcal{L}}\alpha_i|i_L\rangle\in\mathbb{C}^{l\otimes n},
\end{align*}
and ${\rm Enc}(|\psi\rangle)=|\Psi\rangle$ is called a quantum code, and the map ${\rm Enc}$ is called the encoder for the code.
Here, vectors $|0_L\rangle,|1_L\rangle,|2_L\rangle,\dots,|l-1_L\rangle\in\mathbb{C}^{l\otimes n}$ are collectively called logical codewords.

\begin{dfn}[Quantum error]\label{dfn2_1}
For positive integers $n,n'$, take a set $\mathcal{A}=\{A_a\}$ of $l^{n'}$-by-$l^n$ matrices such that $\sum_{a}A_a^\dag A_a$ is the identity operator on $S(\mathbb{C}^{l\otimes n})$.
The map $E:S(\mathbb{C}^{l\otimes n})\rightarrow S(\mathbb{C}^{l\otimes n'})$ defined by
\begin{align*}
E(\rho):=\sum_{a}A_a\rho A_a^\dag
\end{align*}
is called a quantum error.
Here, $\mathcal{A}$ is called a Kraus set for the error $E$ and an element of $\mathcal{A}$ is called a Kraus operator.
\end{dfn}

If $n=n'$ and a linear map $E$ is a quantum channel, that is, completely positive and trace-preserving, then $E$ is a quantum error\cite{kraus1983states}.
Definition \ref{dfn2_1} extends the ordinary definition of quantum error to include the case of $n\neq n'$.
This allows the case where the number of particles changes before and after an error, such as quantum insertion or quantum deletion, to be considered a quantum error.

If $|\psi\rangle\in S(\mathbb{C}^l)$ can be obtained by the operations allowed by quantum mechanics on $E({\rm Enc}(|\psi\rangle))\in S(\mathbb{C}^{l\otimes n'})$, then the quantum code $C$ is correctable for the quantum error $E$.
In this paper, the quantum code $C$ is called $\mathcal{A}$-correcting when $|0\dots0\rangle\otimes|\psi\rangle\in S(\mathbb{C}^{l\otimes n'})$ can be obtained by a finite number of measurements and unitary transformations.
When $C$ is an $\mathcal{A}$-correcting code, the original state $|\psi\rangle\in S(\mathbb{C}^l)$ can be obtained by deleting the $1$st through the $(n'-1)$th particles.
The decoding in this paper uses the measurement described by the measurement operator $\mathcal{M}=\{M_k\}$ that satisfies the completeness relation $\sum_{k}M_k^\dag M_k=I$.
If we perform a measurement $\mathcal{M}$ under a state $\rho\in S(\mathbb{C}^{l\otimes n'})$, the probability to get an outcome $k$ is $p(k)={\rm Tr}(M_k^\dag M_k\rho)$ and the after measurement state is $M_k\rho M_k^\dag/p(k)$.
Here, we denote the sum of the diagonal elements of the square matrix $M$ by ${\rm Tr}(M)$ and the identity operator by $I$.

%%%%%%%%%%%%%%%%%%%%%%%%%%%%%%%%%%%%%%%%%%%%%%%%%%%%%%%%%%%%%
\section{Correctable condition for quantum errors with non-square Kraus operators}\label{sec3}

This paper extends the recovery superoperator given in Ref.\cite{knill1997} to apply to non-square Kraus operators.
\begin{dfn}[Recovery superoperator]\label{dfn2_2}
The set $\mathcal{R}=\{R_r\}$ of square matrices of order $l^{n'}$ such that the following two conditions are satisfied is called a recovery superoperator of $(C,\mathcal{A})$.
\begin{itemize}
\item For any $i\in\mathcal{L}$, $R_r\in\mathcal{R}$, $A_a\in\mathcal{A}$, there exists $\lambda_{r,a}\in\mathbb{C}$ such that 
\begin{align*}
R_rA_a|i_L\rangle=\lambda_{r,a}|0\dots0i\rangle.
\end{align*}
\item $\mathcal{R}$ satisfies the completeness relation $\sum_{r}R_r^\dag R_r=I$.
\end{itemize}
\end{dfn}

In Theorem \ref{thm3_1}, the case $n=n'$ is a well-known fact, and the 3rd condition in the theorem is known as KL condition\cite{knill1997}.
Here we extend it and claim that it also holds for the case of non-square Kraus operators.

\begin{thm}\label{thm3_1}
For a quantum code $C$ with logical codewords $|0_L\rangle,|1_L\rangle,|2_L\rangle,\dots,|l-1_L\rangle\in\mathbb{C}^{l\otimes n}$ and a quantum error $E:S(\mathbb{C}^{l\otimes n})\rightarrow S(\mathbb{C}^{l\otimes n'})$ with the Kraus set $\mathcal{A}$, the following three conditions are equivalent to each other:
\begin{itemize}
\item[${\it 1.}$] The quantum code $C$ is $\mathcal{A}$-correcting.
\item[${\it 2.}$] There exists a recovery superoperator of $(C,\mathcal{A})$.
\item[${\it 3.}$] (KL condition) For any $A_a,A_b\in\mathcal{A}$, and any $i,j\in\mathcal{L}$, there exists $\mu_{a,b}\in\mathbb{C}$ such that
\begin{align*}
\langle i_L|A_a^\dag A_b|j_L\rangle=\mu_{a,b}\delta_{i,j},
\end{align*}
where $\delta_{i,j}$ denotes the Kronecker delta function.
\end{itemize}
\end{thm}

In the case of non-square Kraus operators, it may not be possible to correct to the post-encoding state because the number of particles has changed, but it is possible to correct to the pre-encoding state.
The concept of the proof of Theorem \ref{thm3_1} is based on the one by Knill-Laflamme\cite{knill1997}, but needs some modifications in this sense.

\begin{proof}
This proof consists of three steps.

\noindent \textit{(Step 1)}
First, we prove that ${\it 1}\Rightarrow{\it 2}$.

Assume that the outcomes obtained in the decoding process are in turn $k_1,k_2,\dots,k_m$, and the operators corresponding to them are in turn
\begin{align*}
U_0,\mathcal{M}^1,U_1,\mathcal{M}^2,U_2,\dots,\mathcal{M}^m,U_m,
\end{align*}
where, $U_0,U_1,U_2,\dots,U_m$ are unitary operators and $\mathcal{M}^1=\{M^1_{k_1}\},\mathcal{M}^2=\{M^2_{k_2}\},\dots,\mathcal{M}^m=\{M^m_{k_m}\}$ are measurement operators.
Since the product of unitary matrices is a unitary matrix and the identity matrix is a unitary matrix, we may assume that the measurement and unitary transformation are performed alternately as described above.
Note that the operator used may vary depending on the outcomes obtained each time.
In this case, define $R_{\bm{k}}=U_{m}M_{k_m}^m\dots U_2M_{k_2}^2U_1M_{k_1}^1U_0$ with $\bm{k}=k_1k_2\dots k_m$, we can show that $\mathcal{R}=\{R_{\bm{k}}\}$ is a recovery superoperator.
Since 
\begin{align*}
\lefteqn{R_{\bm{k}}\left(\sum_aA_a\left(\sum_{i\in\mathcal{L}}\alpha_i|i_L\rangle\right)\left(\sum_{i\in\mathcal{L}}\overline{\alpha_i}\langle i_L|\right)A_a^\dag\right)R_{\bm{k}}^\dag}\\
&=\prod_{j=1}^mp_j(k_j)\left(\sum_{i\in\mathcal{L}}\alpha_i|0\dots0i\rangle\right)\left(\sum_{i\in\mathcal{L}}\overline{\alpha_i}\langle0\dots0i|\right)
\end{align*}
for any $|\psi\rangle=\sum_{i\in\mathcal{L}}\alpha_i|i\rangle\in\mathbb{C}^l$, we have
\begin{align*}
R_{\bm{k}}A_a|i_L\rangle=\sqrt{\prod_{j=1}^mp_j(k_j)}|0\dots0i\rangle,
\end{align*}
where $p_j(k_j)$ is the probability of obtaining the outcome $k_j$ in the measurement $\mathcal{M}^j$ for $1\leq j\leq m$.
We can also check that $\sum_{\bm{k}}R_{\bm{k}}^\dag R_{\bm{k}}=I$ by considering that $\sum_{M_k\in\mathcal{M}}M_k^\dag M_k=I$ for the same measurement $\mathcal{M}$ and $U^\dag U=I$ for any unitary matrix $U$.

\noindent \textit{(Step 2)}
Next, we prove that ${\it 2}\Rightarrow{\it 3}$.

For any $A_a,A_b\in\mathcal{A}$, and any $i,j\in\mathcal{L}$, we obtain
\begin{align*}
\langle i_L|A_a^\dag A_b|j_L\rangle
&=\langle i_L|A_a^\dag \left({\sum_rR_r^\dag R_r}\right)A_b|j_L\rangle\\
&=\sum_r\langle i_L|A_a^\dag R_r^\dag R_rA_b|j_L\rangle\\
&=\sum_r\overline{\lambda_{r,a}}\lambda_{r,b}\langle0\dots0i|0\dots0j\rangle\\
&=\left(\sum_r\overline{\lambda_{r,a}}\lambda_{r,b}\right)\delta_{i,j}.
\end{align*}
Thus, KL condition is satisfied.

\noindent \textit{(Step 3)}
Finally, we prove that ${\it 3}\Rightarrow{\it 1}$.

Fix $i\in\mathcal{L}$ and let $V^i\subset\mathbb{C}^{l^{n'}}$ be the vector space spanned by $A_a|i_L\rangle$ for all Kraus operators $A_a\in\mathcal{A}$.
We define the basis $\{|u_1^i\rangle,|u_2^i\rangle,\dots,|u_d^i\rangle\}$ of $V^i$ by applying the Gram-Schmidt orthonormalization as follows.
We define
\begin{align*}
|\tilde{u}_1^i\rangle&:=A_1|i_L\rangle,\\
|\tilde{u}_k^i\rangle&:=A_k|i_L\rangle-\sum_{p=1}^{k-1}\frac{\langle\tilde{u}_p^i|A_k|i_L\rangle}{\langle\tilde{u}_p^i|\tilde{u}_p^i\rangle}|\tilde{u}_p^i\rangle
\end{align*}
for $k\geq2$ to obtain vectors $|\tilde{u}_1^i\rangle, |\tilde{u}_2^i\rangle,\dots,|\tilde{u}_d^i\rangle\in V^i$ that are orthogonal to each other.
Note that even if the Kraus set $\mathcal{A}$ is an infinite set, the above inductive operation must finish because the dimension $d$ of $V^i$ is finite.
In addition, even if the Kraus set $\mathcal{A}$ is uncountable, the orthogonal vectors $|\tilde{u}_1^i\rangle, |\tilde{u}_2^i\rangle,\dots,|\tilde{u}_d^i\rangle\in V^i$ can be obtained by selecting the Kraus operators $A_1,A_2,\dots,A_d$ in any order such that $A_1|i_L\rangle,A_2|i_L\rangle,\dots,A_d|i_L\rangle$ are linearly independent.
By adjusting the lengths of all the obtained orthogonal vectors $|\tilde{u}_1^i\rangle, |\tilde{u}_2^i\rangle,\dots,|\tilde{u}_d^i\rangle\in V^i$ to $1$, the basis $\{|u_1^i\rangle,|u_2^i\rangle,\dots,|u_d^i\rangle\}$ of $V^i$ is obtained.
Note here that $d$ is independent of $i$.

From the definition of the basis $\{|u_1^i\rangle,|u_2^i\rangle,\dots,|u_d^i\rangle\}$, we can represent
\begin{align}
|u_k^i\rangle=c_{k,1}A_1|i_L\rangle+c_{k,2}A_2|i_L\rangle+\dots+c_{k,k}A_k|i_L\rangle\label{eq1}
\end{align}
for $k\in [d]$.
Note that $c_{k,1}, c_{k,2}, \dots, c_{k,k}\in\mathbb{C}$ do not depend on $i$.
This fact can be proved by mathematical induction as follows.
From $\langle\tilde{u}_1^i|\tilde{u}_1^i\rangle=\langle i_L|A_1^\dag A_1|i_L\rangle=\mu_{1,1}$ for any $i\in\mathcal{L}$,
\begin{align*}
|u_1^i\rangle=\frac{|\tilde{u}_1^i\rangle}{\||\tilde{u}_1^i\rangle\|}=\frac{A_1|i_L\rangle}{\sqrt{\mu_{1,1}}}
\end{align*}
holds and $\mu_{1,1}$ does not depend on $i$, thus it holds when $k=1$.
When $k\leq m$ for an integer $m\in [d-1]$, assume that
\begin{align*}
|\tilde{u}_k^i\rangle=c'_{k,1}A_1|i_L\rangle+c'_{k,2}A_2|i_L\rangle+\dots+c'_{k,k}A_k|i_L\rangle
\end{align*}
and $c'_{k,1}, c'_{k,2}, \dots, c'_{k,k}\in\mathbb{C}$ are all independent of $i$.
Then,
\begin{align*}
\langle\tilde{u}_k^i|\tilde{u}_k^i\rangle&=\sum_{p=1}^k\sum_{q=1}^k\overline{c'_{k,p}}c'_{k,q}\langle i_L|A_p^\dag A_q|i_L\rangle\\
&=\sum_{p=1}^k\sum_{q=1}^k\overline{c'_{k,p}}c'_{k,q}\mu_{p,q}
\end{align*}
and
\begin{align*}
\langle\tilde{u}_k^i|A_{m+1}|i_L\rangle&=\sum_{p=1}^k\overline{c'_{k,p}}\langle i_L|A_p^\dag A_{m+1}|i_L\rangle\\
&=\sum_{p=1}^k\overline{c'_{k,p}}\mu_{p,m+1}
\end{align*}
are both independent of $i$.
Therefore, every coefficient of
\begin{align*}
|\tilde{u}_{m+1}^i\rangle&=A_{m+1}|i_L\rangle-\sum_{p=1}^m\frac{\langle\tilde{u}_p^i|A_{m+1}|i_L\rangle}{\langle\tilde{u}_p^i|\tilde{u}_p^i\rangle}|\tilde{u}_p^i\rangle\\
&=A_{m+1}|i_L\rangle-\sum_{p=1}^m\left(\frac{\langle\tilde{u}_p^i|A_{m+1}|i_L\rangle}{\langle\tilde{u}_p^i|\tilde{u}_p^i\rangle}\sum_{q=1}^pc'_{p,q}A_q|i_L\rangle\right)
\end{align*}
does not depend on $i$.
Hence, all coefficients of $|u_{m+1}^i\rangle$ are also independent of $i$, and it is shown that it holds for $k=m+1$.

Define $M_k:=\sum_{i\in\mathcal{L}}|u_k^i\rangle\langle u_k^i|$ for $k\in [d]$.
For the space $V\subset\mathbb{C}^{l\otimes n'}$ with the basis $\{|u_k^i\rangle\mid k\in [d], i\in\mathcal{L}\}$, choose a basis $\{|e_j\rangle\}$ of its orthogonal complementary space $V^\perp$ and define $M_\emptyset:=\sum_{j}|e_j\rangle\langle e_j|$.
Then $\mathcal{M}:=\{M_{\emptyset}\}\cup \{M_{k}\mid k\in [d]\}$ is a set of measurement operators because the completeness relation
\begin{align*}
M_\emptyset^\dag M_\emptyset+\sum_{k\in [d]}M_k^\dag M_k=I
\end{align*}
is satisfied.
Here, $I$ is the identity matrix of order $l^{n'}$.

The measurement $\mathcal{M}$ is performed on the state after the quantum error $E$ of the encoded state $|\Psi\rangle\in S(\mathbb{C}^{l\otimes n})$.
Then, the probability to get the outcome $k\in [d]$ is 
\begin{align*}
p(k)&={\rm Tr}\left(M_k^\dag M_k\left(\sum_a A_a|\Psi\rangle\langle\Psi|A_a^\dag\right)\right)\\
&=\sum_{\bm{x}\in\mathcal{L}^{n'}}\sum_a\langle \bm{x}|M_k^\dag M_kA_a|\Psi\rangle\langle\Psi|A_a^\dag|\bm{x}\rangle\\
&=\sum_a\langle\Psi|A_a^\dag\left(\sum_{\bm{x}\in\mathcal{L}^{n'}}|\bm{x}\rangle\langle \bm{x}|\right)M_k^\dag M_kA_a|\Psi\rangle\\
&=\sum_a\langle\Psi|A_a^\dag M_k^\dag M_kA_a|\Psi\rangle.
\end{align*}
Since $\langle u_k^i|A_a|i_L\rangle$ does not depend on $i$ from Equation (\ref{eq1}), if we set this value to $\beta_{k,a}\in\mathbb{C}$, then
\begin{align*}
\langle u_k^i|A_a|j_L\rangle=
\begin{cases}
\beta_{k,a}&i=j\\
0&i\neq j
\end{cases}
\end{align*}
holds. 
Therefore, we obtain
\begin{align}
M_kA_a|\Psi\rangle&=\sum_{i\in\mathcal{L}}|u_k^i\rangle\langle u_k^i|A_a\sum_{j\in\mathcal{L}}\alpha_j|j_L\rangle\nonumber\\
&=\sum_{i\in\mathcal{L}}\sum_{j\in\mathcal{L}}\alpha_j|u_k^i\rangle\langle u_k^i|A_a|j_L\rangle\nonumber\\
&=\sum_{i\in\mathcal{L}}\alpha_i\beta_{k,a}|u_k^i\rangle\nonumber\\
&=\beta_{k,a}\sum_{i\in\mathcal{L}}\alpha_i|u_k^i\rangle.\label{eq2}
\end{align}
Hence, we have
\begin{align}
p(k)&=\sum_{a}\overline{\beta_{k,a}}\beta_{k,a}\left(\sum_{i\in\mathcal{L}}\overline{\alpha_i}\langle u_k^i|\right)\left(\sum_{i\in\mathcal{L}}\alpha_i|u_k^i\rangle\right)\nonumber\\
&=\sum_a|\beta_{k,a}|^2.\label{eq3}
\end{align}
The state after the measurement $\mathcal{M}$ when the outcome $k$ is obtained is $\sum_{i\in\mathcal{L}}\alpha_i|u_k^i\rangle\in S(\mathbb{C}^{l\otimes n'})$, since
\begin{align*}
M_k\left(\sum_aA_a|\Psi\rangle\langle\Psi|A_a^\dag\right)M_k^\dag
&=\sum_aM_kA_a|\Psi\rangle\langle\Psi|A_a^\dag M_k^\dag\\
&=\sum_a\beta_{k,a}\overline{\beta_{k,a}}\left(\sum_{i\in\mathcal{L}}\alpha_i|u_k^i\rangle\right)\left(\sum_{i\in\mathcal{L}}\overline{\alpha_i}\langle u_k^i|\right)\\
&=p(k)\left(\sum_{i\in\mathcal{L}}\alpha_i|u_k^i\rangle\right)\left(\sum_{i\in\mathcal{L}}\overline{\alpha_i}\langle u_k^i|\right)
\end{align*}
from Equation (\ref{eq2}).
On the other hand, from the definition of $M_\emptyset$, the probability of obtaining the outcome $\emptyset$ is $0$.

For each $k\in [d]$, take one unitary matrix $U_k$ such that $U_k|u_k^i\rangle=|0\dots0i\rangle$ for any $i\in\mathcal{L}$.
When the outcome $k$ is obtained, applying the unitary operator $U_k$ to the state after the measurement $\mathcal{M}$, we obtain
\begin{align*}
U_k\left(\sum_{i\in\mathcal{L}}\alpha_i|u_k^i\rangle\right)=\sum_{i\in\mathcal{L}}\alpha_i|0\dots0i\rangle=|0\dots0\rangle\otimes|\psi\rangle.
\end{align*}
Therefore, the code $C$ is $\mathcal{A}$-correcting.
The original state $|\psi\rangle=\sum_{i\in\mathcal{L}}\alpha_i|i\rangle\in S(\mathbb{C}^l)$ can be obtained by deleting the $1$st through the $(n'-1)$th particles and error-correction is completed.
\end{proof}

%%%%%%%%%%%%%%%%%%%%%%%%%%%%%%%%%%%%%%%%%%%%%%%%%%%%%%%%%%%%%
\section{Equivalence of quantum single deletion and single insertion error-correctabilities}\label{sec4}

As typical errors that change the number of particles, this section defines single deletion errors and single insertion errors using Kraus operators.

Let $m\geq0$ be an integer and let $|\phi\rangle\in\mathbb{C}^l$ with $\langle\phi|\phi\rangle=1$.
For an integer $p\in [m+1]$, we define a $l^m$-by-$l^{m+1}$ matrix $D_{p,|\phi\rangle}^m$ and a $l^{m+1}$-by-$l^m$ matrix $I_{p,|\phi\rangle}^m$ as 
\begin{align*}
D_{p,|\phi\rangle}^m&:=\underbrace{\mathbb{I}_l\otimes\dots\otimes\mathbb{I}_l}_{(p-1)\,{\rm times}}\otimes\langle\phi|\otimes \underbrace{\mathbb{I}_l\otimes\dots\otimes\mathbb{I}_l}_{(m-p+1)\,{\rm times}},\\
I_{p,|\phi\rangle}^m&:=\underbrace{\mathbb{I}_l\otimes\dots\otimes\mathbb{I}_l}_{(p-1)\,{\rm times}}\otimes|\phi\rangle\otimes \underbrace{\mathbb{I}_l\otimes\dots\otimes\mathbb{I}_l}_{(m-p+1)\,{\rm times}}.
\end{align*}
Here, $\mathbb{I}_l$ denotes the identity matrix of order $l$.

\begin{dfn}[Single deletion error]\label{single_del}
For a quantum state $\rho\in S(\mathbb{C}^{l\otimes n})$, we define a quantum single deletion error as a map $E^{del}:S(\mathbb{C}^{l\otimes n})\rightarrow S(\mathbb{C}^{l\otimes (n-1)})$ expressed as
\begin{align*}
E^{del}(\rho):=\sum_{p\in [n]}\left(p^-(p)\sum_{b\in\mathcal{L}}D_{p,|b\rangle}^{n-1}\rho {D_{p,|b\rangle}^{n-1}}^\dag\right),
\end{align*}
with $\sum_{p\in[n]}p^-(p)=1$ for a non-negative-valued function $p^-$.
I.e., the Kraus set for the single deletion error $E^{del}$ is
\begin{align}
\mathcal{D}_1:=\left\{\sqrt{p^-(p)}D_{p,|b\rangle}^{n-1}\,\middle|\, p\in [n],b\in\mathcal{L}\right\}.\label{eqdel}
\end{align}
\end{dfn}
Although quantum single deletion error is often defined as a partial trace at an unknown position, it can also be defined as in Definition \ref{single_del} from the discussion in Ref. \cite{9611450}.

\begin{dfn}[Single insertion error]\label{single_ins}
Suppose that a single qudit $\sigma\in S(\mathbb{C}^{l})$ is represented as $\sigma=\sum_{b\in\mathcal{L}}p_b|\phi_b\rangle\langle\phi_b|$ in the form of a spectral decomposition.
For a quantum state $\rho\in S(\mathbb{C}^{l\otimes n})$, we define a quantum single insertion error of $\sigma$ as a map $E^{ins}_\sigma:S(\mathbb{C}^{l\otimes n})\rightarrow S(\mathbb{C}^{l\otimes (n+1)})$ expressed as
\begin{align*}
E^{ins}_{\sigma}(\rho):=\sum_{p\in [n+1]}\left(p^+(p)\sum_{b\in\mathcal{L}}p_bI_{p,U|b\rangle}^{n}\rho {I_{p,U|b\rangle}^{n}}^\dag\right),
\end{align*}
with $\sum_{p\in[n+1]}p^+(p)=1$ for a non-negative-valued function $p^+$.
Here, $U$ is a unitary matrix such that $U|b\rangle=|\phi_b\rangle$ for every $b\in\mathcal{L}$.
I.e., the Kraus set for the single insertion error $E^{ins}_{\sigma}$ is
\begin{align}
\mathcal{I}_1:=\left\{\sqrt{p^+(p)p_b}I_{p,U|b\rangle}^{n}\,\middle|\, p\in [n+1],b\in\mathcal{L}\right\}.\label{eqins}
\end{align}
\end{dfn}
Quantum single insertion error is often defined as the insertion of a quantum state into an unknown position, but it is known that if the state before the error is pure, it can be expressed in a Kraus representation as in Definition \ref{single_ins}\cite{9834635}.
Throughout this paper, the discussion is based on the assumption that the quantum state before the error is pure.
This means that the codeword is pure, which is a weak assumption.

It is pointed out that the following theorem presented by the author in 2022 does not show the equivalence of error-correctability\cite{9834635}.
\begin{thm}\label{thm:equiv_ins_del}[Theorem III.4,\cite{9834635}]
KL condition for single deletion error is equivalent to that for single insertion error.
\end{thm}
 
From Theorem \ref{thm3_1}, reported for the first time in this paper, the equivalence of error-correctability is also shown.
\begin{thm}\label{thm4_4}
It is equivalent for quantum code to be $\mathcal{D}_1$-correcting and $\mathcal{I}_1$-correcting.
\end{thm}

%%%%%%%%%%%%%%%%%%%%%%%%%%%%%%%%%%%%%%%%%%%%%%%%%%%%%%%%%%%%%
\section{A necessary and sufficient condition for single qudit insertion/deletion error-correction}\label{sec5}

This section discusses the code $C$ defined by the following logical codewords.
Take $A_i\subset\mathcal{L}^n$ for every $i\in\mathcal{L}$ and define
\begin{align*}
|i_L\rangle:=\frac{1}{\sqrt{|A_i|}}\sum_{{\bm a}\in A_i}|\bm{a}\rangle.
\end{align*}

A necessary and sufficient condition for the correction of single deletion errors in the code $C$ is as follows:
\begin{thm}\label{GNH}
It is a necessary and sufficient condition for the code $C$ to be $\mathcal{D}_1$-correcting that the following two conditions are satisfied simultaneously.
\begin{itemize}
\setlength{\leftskip}{-0.2cm}
\item (C1${}^{del}$: ratio condition): For each $p_1,p_2\in [n]$ and $b_1,b_2\in\mathcal{L}$, the following values are constant regardless of $i\in\mathcal{L}$:
\begin{align*}
{|\Delta^-_{p_1,b_1}(A_i)\cap\Delta^-_{p_2,b_2}(A_i)|}\,/\,{|A_i|}.
\end{align*}

\item (C2${}^{del}$: distance condition): For any $p_1,p_2\in [n]$ and any $b_1,b_2,i,j\in\mathcal{L}$, if $i\neq j$, then
\begin{align*}
|\Delta^-_{p_1,b_1}(A_i)\cap\Delta^-_{p_2,b_2}(A_j)|=0.
\end{align*}
\end{itemize}
We refer to (C1${}^{del}$) and (C2${}^{del}$) collectively as (C${}^{del}$).
Here,
\begin{align*}
\Delta_{p,b}^-(A):=\{a_1 \dots a_{p-1} a_{p+1} 
 \dots a_n \mid a_1 \dots a_{p-1} b a_{p+1} \dots a_n \in A \}
\end{align*}
for a non-empty set $A \subset\mathcal{L}^n$ and $p \in [n]$, $b\in\mathcal{L}$.
\end{thm}

\begin{proof}
For $p\in [n]$ and $b,i\in\mathcal{L}$, we have 
\begin{align*}
D_{p,|b\rangle}^{n-1}|i_L\rangle&=\frac{1}{\sqrt{|A_i|}}\sum_{\bm{a}\in A_i}D_{p,|b\rangle}^{n-1}|\bm{a}\rangle=\frac{1}{\sqrt{|A_i|}}\sum_{\tilde{\bm{a}}\in\Delta^-_{p,b}(A_i)}|\tilde{\bm{a}}\rangle.
\end{align*}
Hence, for $p_1,p_2\in[n]$ and $b_1,b_2,i,j\in\mathcal{L}$, we obtain
\begin{align*}
\langle i_L|{D_{p_1,|b_1\rangle}^{n-1}}^\dag D_{p_2,|b_2\rangle}^{n-1}|j_L\rangle&=\frac{|\Delta^-_{p_1,b_1}(A_i)\cap\Delta^-_{p_2,b_2}(A_j)|}{\sqrt{|A_i||A_j|}}.
\end{align*}
Therefore, from Theorem \ref{thm3_1}, the condition (C${}^{del}$) is a necessary and sufficient condition for $\mathcal{D}_1$-correcting.
\end{proof}

Theorem \ref{GNH} extends the code by Nakayama-Hagiwara\cite{Nakayama20202} with $l=2$ to the case where $l$ is any positive integer, and further generalizes it to be necessary and sufficient level of error-correction.
Note that the decoder shown in Ref.\cite{Nakayama20202} is consequently the same as the decoder shown in the proof of Theorem \ref{thm3_1} in this paper.

A necessary and sufficient condition for the correction of single insertion errors in the code $C$ is as follows:
\begin{thm}\label{GSH}
It is a necessary and sufficient condition for the code $C$ to be $\mathcal{I}_1$-correcting that the following two conditions are satisfied simultaneously.
\begin{itemize}
\setlength{\leftskip}{-0.2cm}
\item (C1${}^{ins}$: ratio condition): For each $p_1,p_2\in [n+1]$ and $b_1,b_2\in\mathcal{L}$, the following values are constant regardless of $i\in\mathcal{L}$:
\begin{align*}
|\Delta^+_{p_1,b_1}(A_i)\cap\Delta^+_{p_2,b_2}(A_i)|\,/\,|A_i|.
\end{align*}
\item (C2${}^{ins}$: distance condition): For any $p_1,p_2\in [n+1]$ and any $b_1,b_2,i,j\in\mathcal{L}$, if $i\neq j$, then
\begin{align*}
|\Delta^+_{p_1,b_1}(A_i)\cap\Delta^+_{p_2,b_2}(A_j)|=0.
\end{align*}
\end{itemize}
We refer to (C1${}^{ins}$) and (C2${}^{ins}$) collectively as (C${}^{ins}$).
Here,
\begin{align*}
\Delta_{p,b}^+(A)
 \coloneqq & \{a_1 \dots a_{p-1} b a_{p} \dots a_n 
 \mid a_1 \dots a_{p-1} a_p \dots a_n \in A \}
\end{align*}
for a non-empty set $A \subset\mathcal{L}^n$ and $p \in [n+1]$, $b\in\mathcal{L}$.
\end{thm}

\begin{proof}
For $p\in [n+1]$ and $b,i\in\mathcal{L}$, we have 
\begin{align*}
I_{p,|b\rangle}^{n}|i_L\rangle&=\frac{1}{\sqrt{|A_i|}}\sum_{\bm{a}\in A_i}I_{p,|b\rangle}^{n}|\bm{a}\rangle=\frac{1}{\sqrt{|A_i|}}\sum_{\tilde{\bm{a}}\in\Delta^+_{p,b}(A_i)}|\tilde{\bm{a}}\rangle.
\end{align*}
Hence, for $p_1,p_2\in[n+1]$ and $b_1,b_2,i,j\in\mathcal{L}$, by noting that there exists a unitary matrix $U_{p_1,p_2}$ such that $I_{p_1,U|b_1\rangle}^n=U_{p_1,p_2}I_{p_1,|b_1\rangle}^n$ and $I_{p_2,U|b_2\rangle}^n=U_{p_1,p_2}I_{p_2,|b_2\rangle}^n$, we obtain
\begin{align*}
\langle i_L|{I_{p_1,U|b_1\rangle}^{n}}^\dag I_{p_2,U|b_2\rangle}^{n}|j_L\rangle&=\langle i_L|{I_{p_1,|b_1\rangle}^{n}}^\dag I_{p_2,|b_2\rangle}^{n}|j_L\rangle\\
&=\frac{|\Delta^+_{p_1,b_1}(A_i)\cap\Delta^+_{p_2,b_2}(A_j)|}{\sqrt{|A_i||A_j|}}.
\end{align*}
Therefore, from Theorem \ref{thm3_1}, the condition (C${}^{ins}$) is a necessary and sufficient condition for $\mathcal{I}_1$-correcting.
\end{proof}

Theorem \ref{GSH} extends $l$ to any positive integer for codes with $l=2$ in Ref.\cite{9834635} and greatly generalizes the error-correction condition to necessary and sufficient level.
Note that although the 4-qubit code satisfies the condition (C${}^{ins}$), the decoder given in the proof of Theorem \ref{thm3_1} in this paper is not the one by Hagiwara\cite{ManabuHagiwara20212020XBL0191}, but corresponds to the one given in Ref.\cite{10439039}.

\begin{thm}\label{insdel}
Conditions (C${}^{del}$) and (C${}^{ins}$) are equivalent.
\end{thm}

Theorem \ref{insdel} follows immediately from Theorem \ref{GNH}, Theorem \ref{GSH} and Theorem \ref{thm4_4}, but can also be proved classically as follows:

\begin{proof}[Classical proof of Theorem \ref{insdel}]
For integers $p_1,p_2\in[n]$, assume that $p_1\leq p_2$ without loss of generality.
The sequence $\bm{x}\in\Delta_{p_1,b_1}^-(A)\cap\Delta_{p_2,b_2}^-(A)$ and the sequence $\tilde{\bm{x}}\in\Delta_{p_1,b_1}^+(A)\cap\Delta_{p_2+1,b_2}^+(A)$ correspond one-to-one, where 
\begin{align*}
\bm{x}&=x_1\dots x_{p_1-1}x_{p_1}\dots x_{p_2-1}x_{p_2+1}\dots x_{n-1},\\
\tilde{\bm{x}}&=x_1\dots x_{p_1-1}b_1x_{p_1}\dots x_{p_2-1}b_2x_{p_2+1}\dots x_{n-1}.
\end{align*}
Therefore, we obtain
\begin{align*}
|\Delta_{p_1,b_1}^-(A)\cap\Delta_{p_2,b_2}^-(A)|=|\Delta_{p_1,b_1}^+(A)\cap\Delta_{p_2+1,b_2}^+(A)|.
\end{align*}
On the other hand, we have
\begin{align*}
|\Delta_{p,b_1}^+(A)\cap\Delta_{p,b_2}^+(A)|=
\begin{cases}
|A|&b_1=b_2\\
0&b_1\neq b_2
\end{cases}
\end{align*}
for any $p\in[n+1]$.
From the above, it is shown that (C1${}^{del}$) and (C1${}^{ins}$) are equivalent.

From the equivalence of classical insertion codes and classical deletion codes, the equivalence of (C2${}^{del}$) and (C2${}^{ins}$) is obvious.
\end{proof}

%%%%%%%%%%%%%%%%%%%%%%%%%%%%%%%%%%%%%%%%%%%%%%%%%%%%%%%%%%%%%
\section{Example of insertion/deletion qudit code and its decoding algorithm}\label{sec6}

By defining the sets $A_0,A_1$, and $A_2$ as follows, we can construct an example of a $6$-qudit single insertion/deletion error-correcting code with $l=3$:
\begin{align*}
A_0&=\{001122,112200,220011\},\\
A_1&=\{002211,110022,221100\},\\
A_2&=\{001100,112211,220022\}.
\end{align*}
We can easily check that these sets satisfy the condition (C${}^{del}$) by Table \ref{tab1}.
\begin{table}[htbp]
\caption{$\Delta^-_{p,b}(A_i)$ for $p\in [6]$ and $b,i\in\{0,1,2\}$}
\begin{center}

\begin{tabular}{|c|c|c|c|}
\hline
\multirow{2}{*}{$\Delta^-_{p,b}(A_0)$}&\multicolumn{3}{c|}{$A_0=\{001122,112200,220011\}$}\\
\cline{2-4} 
&$b=0$ & $b=1$ & $b=2$\\\hline
$p=1$& $\{01122\}$ & $\{12200\}$ & $\{20011\}$\\\hline
$p=2$& $\{01122\}$ & $\{12200\}$ & $\{20011\}$\\\hline
$p=3$& $\{22011\}$ & $\{00122\}$ & $\{11200\}$\\\hline
$p=4$& $\{22011\}$ & $\{00122\}$ & $\{11200\}$\\\hline
$p=5$& $\{11220\}$ & $\{22001\}$ & $\{00112\}$\\\hline
$p=6$& $\{11220\}$ & $\{22001\}$ & $\{00112\}$\\\hline
\end{tabular}

\medskip

\begin{tabular}{|c|c|c|c|}
\hline
\multirow{2}{*}{$\Delta^-_{p,b}(A_1)$}&\multicolumn{3}{c|}{$A_1=\{002211,110022,221100\}$}\\
\cline{2-4} 
&$b=0$ & $b=1$ & $b=2$\\\hline
$p=1$& $\{02211\}$ & $\{10022\}$ & $\{21100\}$\\\hline
$p=2$& $\{02211\}$ & $\{10022\}$ & $\{21100\}$\\\hline
$p=3$& $\{11022\}$ & $\{22100\}$ & $\{00211\}$\\\hline
$p=4$& $\{11022\}$ & $\{22100\}$ & $\{00211\}$\\\hline
$p=5$& $\{22110\}$ & $\{00221\}$ & $\{11002\}$\\\hline
$p=6$& $\{22110\}$ & $\{00221\}$ & $\{11002\}$\\\hline
\end{tabular}

\medskip

\begin{tabular}{|c|c|c|c|}
\hline
\multirow{2}{*}{$\Delta^-_{p,b}(A_2)$}&\multicolumn{3}{c|}{$A_2=\{001100,112211,220022\}$}\\
\cline{2-4} 
&$b=0$ & $b=1$ & $b=2$\\\hline
$p=1$& $\{01100\}$ & $\{12211\}$ & $\{20022\}$\\\hline
$p=2$& $\{01100\}$ & $\{12211\}$ & $\{20022\}$\\\hline
$p=3$& $\{22022\}$ & $\{00100\}$ & $\{11211\}$\\\hline
$p=4$& $\{22022\}$ & $\{00100\}$ & $\{11211\}$\\\hline
$p=5$& $\{00110\}$ & $\{11221\}$ & $\{22002\}$\\\hline
$p=6$& $\{00110\}$ & $\{11221\}$ & $\{22002\}$\\\hline
\end{tabular}

\label{tab1}
\end{center}
\end{table}
By Theorem \ref{insdel}, these sets also satisfy the condition (C${}^{ins}$).
Therefore, the quantum code whose logical codewords are 
\begin{align*}
|0_L\rangle&:=\frac1{\sqrt3}(|001122\rangle+|112200\rangle+|220011\rangle),\\
|1_L\rangle&:=\frac1{\sqrt3}(|002211\rangle+|110022\rangle+|221100\rangle),\\
|2_L\rangle&:=\frac1{\sqrt3}(|001100\rangle+|112211\rangle+|220022\rangle),
\end{align*}
is $\mathcal{D}_1$-correcting and $\mathcal{I}_1$-correcting.
Using the decoding algorithm presented under Theorem \ref{thm3_1}, we can correct any single deletion error or any single insertion error as follows.

\subsection{Decoding algorithm for single deletion errors}
Since $n=6$, from Equation (\ref{eqdel}) there are 18 Kraus operators for the single deletion error, which are labeled newly as $A_a$, where $a\in\{1,2,\dots,18\}$.
The dimension of the vector space $V^i$ spanned by $A_a|i_L\rangle$ for all Kraus operators $A_a$ is $9$, and its basis $\{|u_1^i\rangle,|u_2^i\rangle,\dots,|u_9^i\rangle\}$ is defined by applying the Gram-Schmidt orthonormalization.
For example, for $i=0$, we have the following results.
Compare with Table \ref{tab1}.
\begin{align*}
|u_1^0\rangle=|01122\rangle,\\
|u_2^0\rangle=|22011\rangle,\\
|u_3^0\rangle=|11220\rangle, \\
|u_4^0\rangle=|12200\rangle,\\
|u_5^0\rangle=|00122\rangle,\\
|u_6^0\rangle=|22001\rangle,\\
|u_7^0\rangle=|20011\rangle,\\
|u_8^0\rangle=|11200\rangle,\\
|u_9^0\rangle=|00112\rangle.
\end{align*}

When the measurement $\mathcal{M}$ is performed on the quantum state $E^{del}(\alpha_0|0_L\rangle+\alpha_1|1_L\rangle+\alpha_2|2_L\rangle)\in S(\mathbb{C}^{3\otimes5})$, the probability of obtaining the outcome $k\in[9]$ is
\begin{align*}
p(k)=
\begin{cases}
\displaystyle\frac13(p^-(1)+p^-(2))&k=1,4,7\medskip\\
\displaystyle\frac13(p^-(3)+p^-(4))&k=2,5,8\medskip\\
\displaystyle\frac13(p^-(5)+p^-(6))&k=3,6,9
\end{cases}
\end{align*}
from Equation (\ref{eq3}).
Then, the unitary operator $U_k$ corresponding to each outcome $k\in[9]$ is applied and the $1$st through $4$th particles are deleted.
Thus, the original quantum state $\alpha_0|0\rangle+\alpha_1|1\rangle+\alpha_2|2\rangle)\in S(\mathbb{C}^3)$ can be obtained.

\subsection{Decoding algorithm for single insertion errors}

From Equation (\ref{eqins}) there are $21$ Kraus operators for the single insertion error, which are labeled newly as $A_a$, where $a\in\{1,2,\dots,21\}$.
The dimension of the vector space $V^i$ spanned by $A_a|i_L\rangle$ for all Kraus operators $A_a$ is $21$, and its basis $\{|u_1^i\rangle,|u_2^i\rangle,\dots,|u_{21}^i\rangle\}$ is defined by applying the Gram-Schmidt orthonormalization.
We have
\begin{align*}
|u_{7j+1}^i\rangle&=&&\!\!\!\!I_{1,|j\rangle}^6|i_L\rangle,\\
|u_{7j+2}^i\rangle&=&&\!\!\!\!\sqrt{\frac{9}{8}}\bigg(-\frac13I_{1,|j\rangle}^6|i_L\rangle+I_{2,|j\rangle}^6|i_L\rangle\bigg),\\
|u_{7j+3}^i\rangle&=&&\!\!\!\!\sqrt{\frac{6}{5}}\bigg(-\frac14I_{1,|j\rangle}^6|i_L\rangle-\frac14I_{2,|j\rangle}^6|i_L\rangle+I_{3,|j\rangle}^6|i_L\rangle\bigg),\\
|u_{7j+4}^i\rangle&=&&\!\!\!\!\sqrt{\frac{15}{13}}\bigg(\frac1{10}I_{1,|j\rangle}^6|i_L\rangle+\frac1{10}I_{2,|j\rangle}^6|i_L\rangle-\frac2{5}I_{3,|j\rangle}^6|i_L\rangle+I_{4,|j\rangle}^6|i_L\rangle\bigg),\\
|u_{7j+5}^i\rangle&=&&\!\!\!\!\sqrt{\frac{39}{32}}\bigg(\frac1{13}I_{1,|j\rangle}^6|i_L\rangle+\frac1{13}I_{2,|j\rangle}^6|i_L\rangle-\frac4{13}I_{3,|j\rangle}^6|i_L\rangle-\frac3{13}I_{4,|j\rangle}^6|i_L\rangle+I_{5,|j\rangle}^6|i_L\rangle\bigg),\\
|u_{7j+6}^i\rangle&=&&\!\!\!\!\sqrt{\frac{96}{83}}\bigg(-\frac1{32}I_{1,|j\rangle}^6|i_L\rangle-\frac1{32}I_{2,|j\rangle}^6|i_L\rangle+\frac1{8}I_{3,|j\rangle}^6|i_L\rangle+\frac3{32}I_{4,|j\rangle}^6|i_L\rangle\\
&&&\!\!\!\!-\frac{13}{32}I_{5,|j\rangle}^6|i_L\rangle+I_{6,|j\rangle}^6|i_L\rangle\bigg),\\
|u_{7j+7}^i\rangle&=&&\!\!\!\!\sqrt{\frac{83}{68}}\bigg(-\frac2{83}I_{1,|j\rangle}^6|i_L\rangle-\frac2{83}I_{2,|j\rangle}^6|i_L\rangle+\frac8{83}I_{3,|j\rangle}^6|i_L\rangle+\frac6{83}I_{4,|j\rangle}^6|i_L\rangle\\
&&&\!\!\!\!-\frac{26}{83}I_{5,|j\rangle}^6|i_L\rangle-\frac{19}{83}I_{6,|j\rangle}^6|i_L\rangle+I_{7,|j\rangle}^6|i_L\rangle\bigg),
\end{align*}
for $i,j\in\{0,1,2\}$.
Unlike the deletion case, since $A_a|i_L\rangle$ are not all orthogonal to each other, the result is complicated.

When the measurement $\mathcal{M}$ is performed on the quantum state $E^{ins}(\alpha_0|0_L\rangle+\alpha_1|1_L\rangle+\alpha_2|2_L\rangle)\in S(\mathbb{C}^{3\otimes7})$, the probability of obtaining the outcome $k\in[21]$ is
\begin{align*}
p(k)=
\begin{cases}
\displaystyle p_j\left(p^+(1)+\frac19p^+(2)+\frac19p^+(3)\right)&k=7j+1\medskip\\
\displaystyle p_j\left(\frac89p^+(2)+\frac1{18}p^+(3)\right)&k=7j+2\medskip\\
\displaystyle p_j\left(\frac{5}{6}p^+(3)+\frac{2}{15}p^+(4)+\frac{2}{15}p^+(5)\right)&k=7j+3\medskip\\
\displaystyle p_j\left(\frac{13}{15}p^+(4)+\frac{3}{65}p^+(5)\right)&k=7j+4\medskip\\
\displaystyle p_j\left(\frac{32}{39}p^+(5)+\frac{13}{96}p^+(6)+\frac{13}{96}p^+(7)\right)&k=7j+5\medskip\\
\displaystyle p_j\left(\frac{83}{96}p^+(6)+\frac{361}{7968}p^+(7)\right)&k=7j+6\medskip\\
\displaystyle p_j\left(\frac{68}{83}p^+(7)\right)&k=7j+7
\end{cases}
\end{align*}
for $j\in\{0,1,2\}$ from Equation (\ref{eq3}).
Then, the unitary operator $U_k$ corresponding to each outcome $k\in[21]$ is applied and the $1$st through $6$th particles are deleted.
Thus, the original quantum state $\alpha_0|0\rangle+\alpha_1|1\rangle+\alpha_2|2\rangle)\in S(\mathbb{C}^3)$ can be obtained.

For example, if $\sigma=\frac12|0\rangle\langle0|+\frac13|1\rangle\langle1|+\frac16|2\rangle\langle2|\in S(\mathbb{C}^3)$ is inserted in the $4$th position, that is, if we consider the single insertion error where $p^+(4)=1$, $p^+(p)=0$ for any $p\neq4$, and $p_1=1/2$, $p_2=1/3$, $p_3=1/6$, the error-correction process is as shown in Figure \ref{fig1}.

\begin{figure}[htbp]
\begin{center}
\includegraphics[scale=0.44]{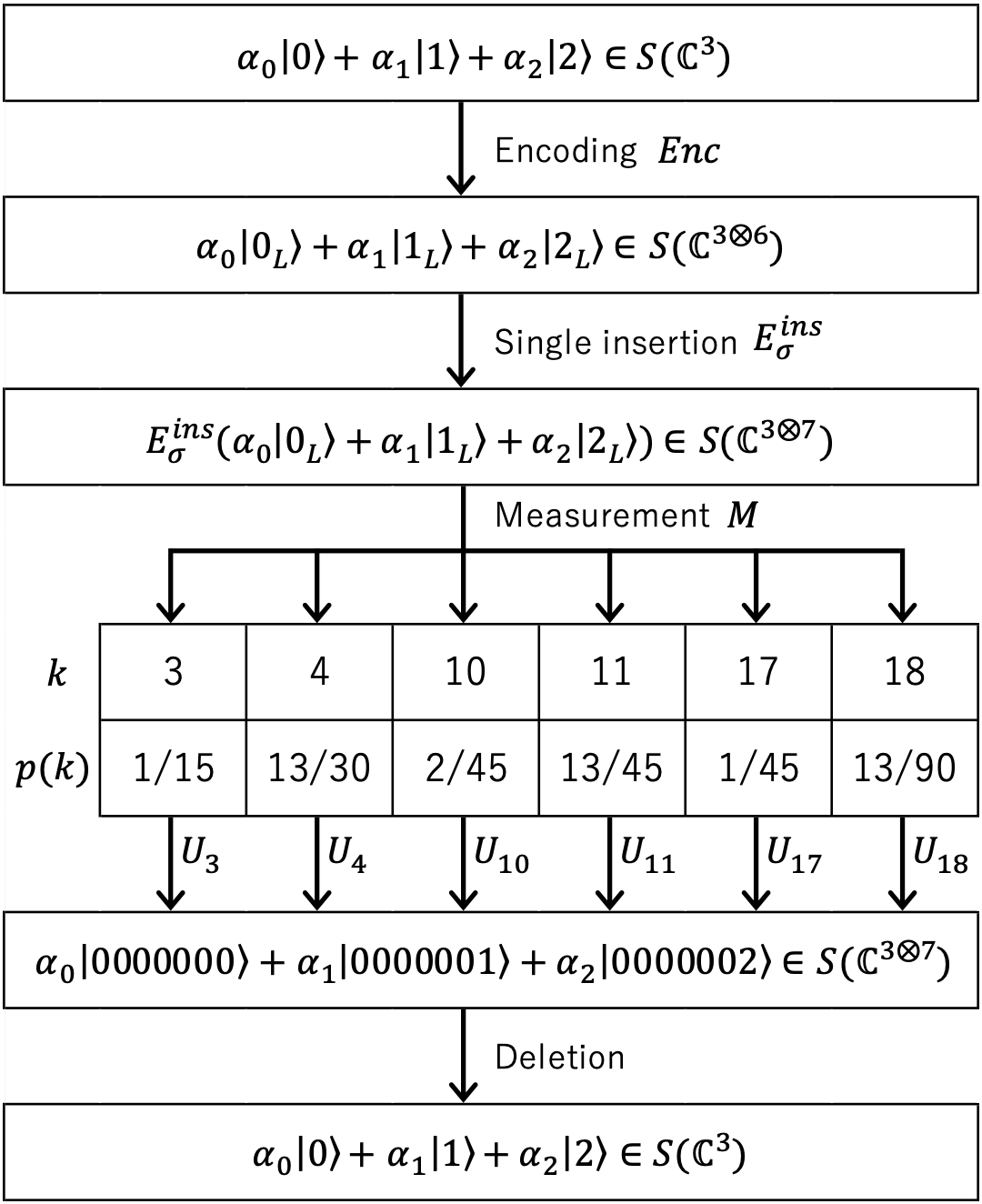}
\end{center}
\caption{Error-correction process for the 6-qudit insertion code when $\sigma=\frac12|0\rangle\langle0|+\frac13|1\rangle\langle1|+\frac16|2\rangle\langle2|\in S(\mathbb{C}^3)$ is inserted in the 4th position}  \label{fig1}
\end{figure}

%%%%%%%%%%%%%%%%%%%%%%%%%%%%%%%%%%%%%%%%%%%%%%%%%%%%%%%%%%%%%
\section{Conclusion}\label{sec7}
In this paper, we proved that KL condition \cite{knill1997} can be used as a necessary and sufficient condition for correcting quantum errors with changing number of particles.
By using KL condition, we showed that the error-correctabilities of single deletion and single insertion are equivalent.
Furthermore, the error-correctability conditions for single deletion given by Nakayama-Hagiwara\cite{Nakayama20202} and for single insertion given by Shibayama-Hagiwara\cite{9834635} were improved to necessary and sufficient level.
We also constructed a new single qudit insertion/deletion code by giving an example that satisfies the conditions, and provided a decoding algorithm for the code.

\bibliography{main}

\end{document}